\documentclass[letterpaper, 10 pt, conference]{ieeeconf}  
\usepackage{amsmath,amssymb,color}
\usepackage{graphicx}
\usepackage{algorithmic,theorem}
\newcommand{\RNum}[1]{\uppercase\expandafter{\romannumeral #1\relax}}
\usepackage[ruled,vlined,linesnumbered]{algorithm2e}

\IEEEoverridecommandlockouts                              
\newtheorem{theorem}{Theorem}[section]

\newtheorem{definition}{Definition}[section]
\newtheorem{corollary}{Corollary}[section]
\newtheorem{lemma}{Lemma}[section]

\newcommand{\erf}{\operatorname{erf}}
\newcommand{\abs}[1]{|{#1}|}
\newcommand{\area}{\operatorname{Area}}
 
\newcommand{\subscr}[2]{{#1}_{\textup{#2}}}

\title{\LARGE \bf
Dynamic Boundary Guarding Against Radially Incoming Targets\thanks{The authors are with the Electrical and Computer Engineering Department, Michigan State University, East Lansing, MI. Emails: \texttt{bajajshi@msu.edu, shaunak@egr.msu.edu}}
}

\author{Shivam Bajaj and Shaunak D. Bopardikar
}

\begin{document}

\maketitle
\thispagestyle{empty}
\pagestyle{empty}

\begin{abstract}
We introduce a dynamic vehicle routing problem in which a single vehicle seeks to guard a circular perimeter against radially inward moving targets. Targets are generated uniformly as per a Poisson process in time with a fixed arrival rate on the boundary of a circle with a larger radius and concentric with the perimeter. Upon generation, each target moves radially inward toward the perimeter with a fixed speed. The aim of the vehicle is to maximize the capture fraction, i.e., the fraction of targets intercepted before they enter the perimeter. We first obtain a fundamental upper bound on the capture fraction which is independent of any policy followed by the vehicle. We analyze several policies in the low and high arrival rates of target generation. For low arrival, we propose and analyze a First-Come-First-Served and a Look-Ahead policy based on repeated computation of the path that passes through maximum number of unintercepted targets. For high arrival, we design and analyze a policy based on repeated computation of Euclidean Minimum Hamiltonian path through a fraction of existing targets and show that it is within a constant factor of the optimal. Finally, we provide a numerical study of the performance of the policies in parameter regimes beyond the scope of the analysis.
\end{abstract}

\section{Introduction} \label{introduction}

 Dynamic vehicle routing (DVR) problems are vehicle routing problems in which the vehicles have to plan their paths through the points of interest which arrive sequentially over time. This paper addresses a DVR problem involving moving targets. Targets are generated at the boundary of an environment and move with a constant speed in order to enter a perimeter. A single vehicle is assigned a task of capturing as many targets as it can before they enter the perimeter. This setup is highly relevant in surveillance applications that involve gathering additional information on mobile targets and in civilian space applications such as protecting the space-stations from  debris. Additional applications of this setup are envisioned in guarding of airport runways from rogue drones.

\subsection{Related Work}

Standard vehicle routing problems in operations research are concerned with planning optimal vehicle routes to visit a set of fixed targets. This requires solving an underlying combinatorial optimization problem~\cite{toth2002}. In contrast, DVR requires that the vehicle routes be re-planned as new information becomes available over time which was originally introduced on graphs in~\cite{psaraftis1988dynamic}. Fundamental limits, novel policies and their constant factor optimality guarantees in continuous environments were established in~\cite{bertsimas1993stochastic}. Environments may be dynamically varying~\cite{papastavrou1996, pavone2011adaptive} and the targets may have multiple levels of priorities \cite{smith2010dynamic} or can be randomly recalled~\cite{Bopardikar2020}. The vehicles may be tasked with performing pickup and delivery operations~\cite{waisanen2008dynamic},\cite{treleaven2013asymptotically}; may possess motion constraints~\cite{savla2008traveling, savla2009traveling, itani2008dynamic} and need not require mutual communication~\cite{arsie2009efficient}. We refer the reader to \cite{bullo2011dynamic} for a review of this literature.

\begin{figure}[t]
	\centering
	\includegraphics[scale=0.6]{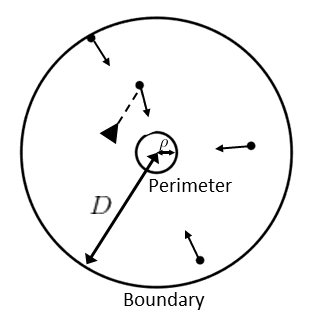}
	\caption{Problem setup. The environment is an annulus of inner radius $\rho$ and outer radius $D$. Targets are depicted as black dots approaching the perimeter at constant speed. The vehicle is shown as a triangle.}
	\label{fig:problem}
\end{figure}
There has been a recent research on pursuit of mobile targets that seek to reach a destination. In \cite{Kawecki2009GuardingAL}, the authors consider a setup of guarding a line segment in the case of a single pursuer and a single evader. They provide with different conditions in which either the defender or the attacker wins. The authors of  \cite{fuchs2018two} consider a pursuit evasion dynamic guarding problem with two vehicles and one demand moving with constant speeds. The vehicles move cooperatively to pursue the demand in a planar environment and gives different cooperative strategy between the two pursuers. The same authors in \cite{inproceedings} consider a differential game of attacker-target-defender. In \cite{Chalasani1999,bopardikar2010dynamic}, the authors consider demands that are slower than the vehicle which is moving parallel to $x$ or $y$ axis. Earlier, we introduced a DVR boundary guarding problem in which a single vehicle was assigned to stop the demands from reaching a deadline in a rectangular environment~\cite{smith2009dynamic}. Our work~\cite{agharkar2015vehicle} considered the case of targets being generated inside a disk-like environment and moving radially outward to escape the region. 

\subsection{Contributions}
The environment considered in this paper is an annulus of inner radius $\rho$ and outer radius $D$. The targets are generated uniformly and randomly on the boundary as per a Poisson process in time with rate $\lambda$. Upon generation, every target moves with a constant velocity $v < 1$ towards the perimeter. A vehicle, modeled as a first-order integrator moving with unit speed, seeks to intercept the targets before they reach the perimeter. The performance of the vehicle is the expected value of the capture fraction, i.e., the fraction of the targets that are intercepted, at steady state.

Our contributions are as follows. We first determine a policy independent upper bound on the capture fraction of the targets and show that it scales as $O(\frac{1+v}{\sqrt{v\lambda \rho}})$. Second, we design three policies for the vehicle and characterize analytic lower bounds on the resulting capture fraction in limiting parameter regimes. In particular, for low arrival rates of the targets, we show that the capture fraction scales as $\Omega(\frac{1}{1+\lambda \rho})$ using a policy based on intercepting the targets as per the first-come-first-served order. We then characterize the performance of a policy with look-ahead based on repeated computation of a path that maximizes the number of targets intercepted in a horizon. Third, in the regime of high arrival rates $\lambda \to +\infty$, we analyze the performance of a policy based on repeated computation of the Euclidean Minimum Hamiltonian Path (EMHP) through the radially moving targets and show that its performance scales as $\Omega(\frac{1-v}{\sqrt{v(1+\sqrt{v})\lambda \rho}})$. In particular, we show that the capture fraction of this policy is within a constant factor of optimal if $v\to 0^+$. Numerical simulations suggest that our analytic bounds generalize beyond the limiting parameter regimes. Although the focus of the problem is similar to the ones in \cite{smith2009dynamic} and in \cite{agharkar2015vehicle}, the geometry of the environment and direction of the targets yield novel results.

\subsection{Organization}
The paper is organized as follows. Section~\ref{sec:problem} comprises the formal problem definition and a summary of background concepts and novel intermediate properties. In Section~\ref{sec:upperbd}, we derive a novel policy independent upper bound on the capture fraction. In Section~\ref{sec:low}, we present two policies suitable for low arrival rate of the targets and we derive novel lower bounds on their performance. In Section~\ref{sec:high}, we present and analyze a policy suitable for high arrival rates. In Section~\ref{sec:numerics}, we present results of numerical simulations. Finally, Section~\ref{sec:conclusion} summarizes this paper and outlines directions for future research.

\section{Problem Formulation and Preliminaries} \label{sec:problem}
We begin with a mathematical description of the problem and provide preliminary properties of the model considered.

\subsection{Mathematical Modeling and Problem Statement}

Consider a disk like environment $\mathcal{E} = \{(r,\theta) : 0 < \rho \leq r \leq D, \, \forall \,  \theta\in[0,2\pi)\}$. Targets appear uniformly and randomly at the boundary of the disk, i.e., at $r = D$. Upon generation, every target moves radially inward towards the inner circular boundary, termed as the $\textit{perimeter}$, having radius $\rho$ in $\mathcal{E}$. The arrival times of the targets are as per a Poisson process with rate $\lambda$. We consider a single service vehicle with motion modeled as a first order integrator with unit speed and with the ability to move in $\mathbb{R}^2$. The vehicle is said to \emph{capture} a target when it is collocated with a target while the target is in $\mathcal{E}$. A target gets removed from the environment if it gets captured by the vehicle. A target is said to \emph{escape} when it reaches the perimeter and is not intercepted by the vehicle. We assume that the speed of the targets is less than that of the vehicle ($v < 1$). Hence, a target must be captured within $(D-\rho)/v$ time units of being generated. We refer to this problem as \emph{RIT problem} for convenience.

\medskip

Let $\mathcal{Q}(t) \subset \mathcal{E}$ denote the set of all outstanding target locations at time $t$. If the $i^{th}$ target that arrives gets captured, then it is placed in $\subscr{\mathcal{Q}}{capt}(t)$ having cardinality $\subscr{n}{capt}(t)$, and is removed from $\mathcal{Q}(t)$. Otherwise, it is placed in $\subscr{\mathcal{Q}}{esc}(t)$ having cardinality $\subscr{n}{esc}(t)$ and removed from $\mathcal{Q}(t)$. Akin to our work in \cite{smith2009dynamic} that formally defined causal nature of policies, we will consider both causal and non-causal policies for the RIT problem, defined as follows.

$\textit{Causal Policy:}$ A causal feedback control policy for a vehicle is a map $\mathcal{P} : \mathcal{E}\times\mathbb{F} \to \mathbb{R}^2$, where $\mathbb{F}(\mathcal{E})$ is the set of finite subsets of $\mathcal{E}$, assigning a commanded velocity to the vehicle as a function of the current state of the system, yielding the kinematic model,
\[
\dot{\mathbf{p}}(t) = \mathcal{P}(\mathbf{p}(t),\mathcal{Q}(t)).
\]
$\textit{Non-causal Policy:}$ In a non-causal feedback control policy, the velocity of the vehicle is a function of current and future states of the system. Although these policies are physically unrealizable, they serve as a means to compare the performance of causal policies against the optimal.

$\textit{Problem Statement:}$ The aim of this paper is to design policies $\mathcal{P}$ that maximize the fraction of targets that are captured $\subscr{\mathbb{F}}{cap}(\mathcal{P})$, termed as $\textit{capture fraction}$, where
\[ \subscr{\mathbb F}{cap}(\mathcal{P}) := \lim_{t\to \infty}\sup\mathbb E \left[\frac {\subscr{n}{capt}(t)}{\subscr{n}{capt}(t)+\subscr{n}{esc}(t)} \right]. \]

In the sequel, we propose different policies that are suitable for low and high arrival rates of the targets. But we first present some preliminary and novel results in the next sub-section which will be used to establish the fundamental limit and analyze the policies.
\subsection{Preliminary Results}
We review a concept related to longest paths on graphs as well as derive some basic results which will be used to obtain bounds on paths through a set of radially moving targets.
\subsubsection{Longest Paths in Directed Acyclic Graphs (DAG)} \label{Longest paths}
A graph $G = (V,E)$ is called a \emph{directed} graph if it consists of a set of vertices $V$ and a set of directed edges $E \subset V\times V$ \cite{smith2009dynamic}. A graph is $\textit{acyclic}$ if its first and the last vertex are not same in the sequence. Finding the longest path, i.e. to find a path that visits maximum number of vertices, is NP-hard to solve as its solution requires the solution of Hamiltonian path problem~\cite{korte2012combinatorial}. However, for a DAG, the longest path problem has an efficient dynamic solution, \cite{Christofides:1975:GTA:1098653} with complexity that scales polynomially with the number of vertices.

\subsubsection{Capturable set}
The \emph{capturable set} comprises the set of locations of all targets that can be captured from a given vehicle location, defined formally as follows.
\begin{definition} A vehicle located at $(x,0)$ can capture targets located in the \emph{capturable set},
\[C(x,v,\rho) := \{(r,\theta) \in \mathcal{E} : r < r_c, \forall \theta \in [0,2\pi)\}\]
where,
\[ r_c(x,v,\rho,\theta) = \min\{D, \rho + v\sqrt{\rho^2+x^2-2x\rho\cos\theta} \}.\]
\end{definition}

If the vehicle located at $(x,0)$ requires time $T$ to intercept a target located at $(r,\theta)$, then $r-vT \geq \rho$. By setting $r_c-vT = \rho$, we obtain the expression for $r_c$. The location $r_c$ corresponds to the location of the targets that the vehicle can capture before they enter the perimeter.

\subsubsection{Distribution of Outstanding Targets}
In this subsection, we derive the distribution of targets in any region of the environment at steady state. Recall that an annular section $A(a,b,\theta_1,\theta_2) \subset\mathbb{R}^2$ is the set $A(a,b,\theta_1,\theta_2):= \{(r,\theta)\mid a\leq r\leq b,\theta \in[\theta_1,\theta_2]\}$.

\begin{lemma}[Outstanding targets]\label{outstanding demands}
Suppose the arrival of targets starts at time 0 and no targets are intercepted in the interval $[0,t]$. Let $\mathcal{Q}$ denote the set of all targets in $[D,D-vt], \forall \theta \in [0,2\pi)$ at time t. Then, given a set $\mathcal{R}(r,\Delta r,\theta,\Delta\theta)$, of infinitesimal area $A$, contained in $\mathcal{Q}$,
\[ \mathbb{P}[\mid\mathcal{R}\cap\mathcal{Q}\mid=n] = \frac{e^{-\bar\lambda A}(\bar\lambda A)^n}{n!}\text{, where } \bar\lambda = \frac{\lambda}{2\pi vr}. \]
Furthermore,
\[ \mathbb{E}[\mid\mathcal{R}\cap\mathcal{Q}\mid]=\bar\lambda r\Delta r \Delta\theta, \, \text{ and } \, \text{Var}[\mid\mathcal{R}\cap\mathcal{Q}\mid]=\bar\lambda r\Delta r \Delta\theta. \] 
\end{lemma}
\begin{proof}Let $\mathcal{R}(r,\Delta r,\theta,\Delta\theta)\subset\mathbb{R}^2$ be an area element where $\Delta$r,$\Delta$$\theta$ are infinitesimally small. Let $\mathcal{Q}$ denote the set of all targets in $[D,D-vt], \, \forall \theta \in [0,2\pi)$. Then, the probability that $\mathcal{R}$ contains $n$ points out of $\mathcal{Q}$ at time $t$ satisfies   
\begin{equation}
\begin{aligned}
\mathbb{P}&[\abs{\mathcal{R} \cap \mathcal{Q}}=n] \\
&= \sum_{i=n}^{\infty}\mathbb{P} [\text{ i targets arrive in }[\frac{D-(r+\Delta r)}{v},\frac{D-r}{v}]\\  
&\times\mathbb{P}[\text{ n of i targets generated in }(\theta,\theta+\Delta\theta)]. \nonumber
\end{aligned}
\end{equation}
Since the generation is uniform in space and Poisson in time,
\[ \mathbb{P} [\text{i targets arrive in }[0,\frac{\Delta r}{v}]]= \frac{e^{(\frac{-\lambda\Delta r}{v})}(\frac{\lambda \Delta r}{v})^i}{i!},\]
and
\begin{equation}
\begin{aligned}
\mathbb{P}&[\text{n of i targets generated in }(0,\Delta\theta)]\\
&={i\choose n}\left(\frac{\Delta\theta}{2\pi}\right)^n\left(1-\frac{\Delta\theta}{2\pi}\right)^{i-n}. \nonumber
\end{aligned}
\end{equation}
Let $\Delta r/v = R$ and $\Delta\theta/2\pi=H$. Thus,
\begin{equation}
\begin{aligned}
\mathbb{P}[\mid\mathcal{R}\cap\mathcal{Q}\mid=n] &=  \sum_{i=n}^{\infty}  \frac{e^{-\lambda R}(\lambda R)^i}{i!}  \frac{(i!)(H)^n(1-H)^{i-n}}{(n!)(i-n)!} \\
&=  \frac{e^{-\bar\lambda A}(\bar\lambda A)^n}{n!},\nonumber
\end{aligned}
\end{equation}
where $\bar\lambda = \lambda/v2\pi r$ and $A$ is the area of the annulus sector.
\end{proof}

This  result establishes that the number of unintercepted targets in an annulus is Poisson distributed uniformly with parameter $\frac{\lambda}{v}\Delta d$, where $\Delta d$ is the difference in the radii of the annulus.
\begin{figure}[t]
\centering
	{\includegraphics[scale=0.6]{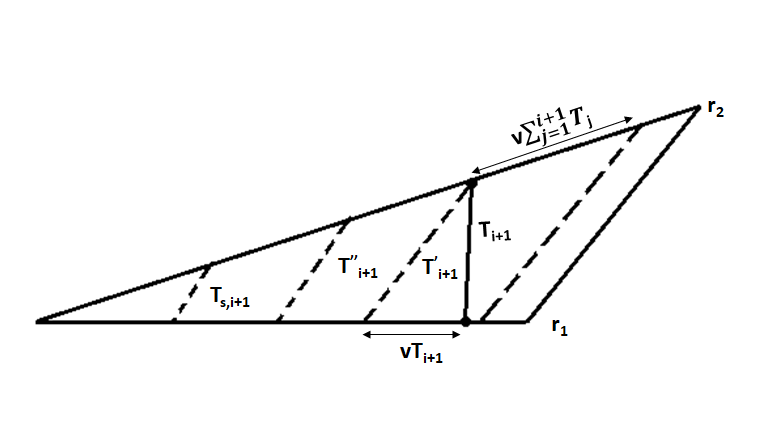}}
	\caption{ The quantity $T_{i+1}$ indicates the time taken by the vehicle to move from target $i$ to the target $i+1$. }
\end{figure}
\subsubsection{Bounds on paths through static and incoming targets}
The following result will be used to bound the length of the path through targets in the environment.
\begin{lemma}[Bounds on path through incoming targets]\label{pathbound}
 Let T be the length of the actual path through the incoming targets and $T_s$ be the length of the path through their static initial locations. Then,
\[ 
\frac{T_s}{1+v} \leq T \leq \frac{T_s}{1-v}. 
\]
\end{lemma}
\begin{proof}Let the targets be labeled in the order in which they arrive. Let $T_j$ be the time taken by the vehicle to capture the $j^{th}$ incoming target. Consider the $i^{th}$ target at $(r_i,\theta_i)$. The vehicle will capture this target at time $\sum_{j=1}^{i}T_j$. It then moves toward the $(i+1)^{th}$ target and reaches it in a time duration of $T_{i+1}$ which is equal to the distance covered due to unit speed. Let the distance between $(r_i-v\sum_{j=1}^{i+1}T_j,\theta_{i})$ and $(r_{i+1}-v\sum_{j=1}^{i+1}T_j,\theta_{i+1})$ be $T^{'}_{i+1}$. Also, let $T_{s,i+1}$ be the distance between $(r_i-v\bar T,\theta_i)$ and $(r_{i+1}-v\bar T,\theta_{i+1})$, where $\bar T<T$. As the distance between two targets moving radially inward with same speed is non-increasing function of time, $T_{s,i+1} \leq T^{'}_{i+1}$. Therefore, from triangle inequality (Fig. 3), $T_{i+1} + vT_{i+1} \geq T^{'}_{i+1}$. This means, $T_{i+1} \geq T^{'}_{i+1}/(1+v) \geq T_{s,i+1}/(1+v)$. Extending this to all the targets we get the expression,
$T=\sum_{i=1}^{n}T_{i+1} \geq \sum_{i=1}^{n}\frac{T_{s,i+1}}{1+v}=\frac{T_s}{1+v}.$
The proof for the upper bound is similar to the proof above and has been omitted for brevity.
\end{proof}
Next we review classic results regarding the upper bound on the length of the shortest path through a set of fixed points in an environment.
\begin{lemma}[Shortest Euclidean path~\cite{few_1955}]\label{Fews}
 Given $n$ points in a square of length $R$, there is a path through the $n$ points of length not exceeding $R\sqrt{2n}+1.75R$.
\end{lemma}

\begin{lemma}[Length of EMHP tour~\cite{BHH_1959}]\label{lem:BHH}
Consider a set $\mathcal{Q}$ of $n$ points independently and uniformly distributed in a compact set $\mathcal{A}$ of area $\abs{\mathcal{A}}$. Then, there exists a constant $\beta_{TSP}$ such that,
\[
\lim_{n\to\infty} \frac{\ell(EMHP(\mathcal{Q}))}{\sqrt{n}}=\beta_{TSP}\sqrt{\abs{\mathcal{A}}},
\]
with probability one. The constant $\beta_{TSP}$ has been estimated numerically as $\beta_{TSP}\approx0.7120\pm0.0002$.
\end{lemma}
\section{A Policy-independent Upper Bound} \label{sec:upperbd}
This section summarizes the first of our analytic results -- a policy independent upper bound on the capture fraction. This is a fundamental limit to this problem and is valid for any values of the problem parameters. First, we derive a lower bound on the expected travel time between two targets, and will be used to derive the fundamental bound.
\begin{lemma}[Travel time lower bound]\label{Expectation time}
If $T_d$ is a random variable denoting the time required to travel between targets in $\mathcal{Q}$, then
\[ \mathbb{E}[T_d] \geq  \frac{1}{1+v}\sqrt \frac{v\pi \rho}{2\lambda}. \]
\end{lemma}
\begin{proof}
Let $S_T$ comprise the set of targets that can be reached from vehicle position $(X,0)$ in at most $T$ time units, as shown in Fig.~\ref{fig:ST}. Mathematically,
\[ 
S_T := \{(r,\theta) \in \mathcal{E} \mid X^2 + (r-v T)^2-2X(r-v T)\cos\theta \leq T^2\}.\]
\begin{figure}[t]
\centering
	{\includegraphics[scale=0.4]{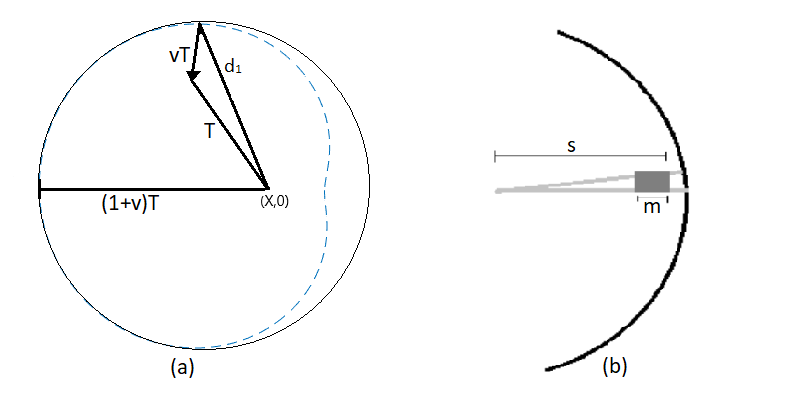}}
	\caption{(a) The set $S_T$ for RIT problem. The set $\bar S_T$ is shown by a solid boundary, which is a circle of radius $(1+v)T$. (b) The area element $\zeta$ of length and width $m$ in $\bar S_T$.}
	\label{fig:ST}
\end{figure}
Also, let
\[  
\bar S_T := \{(r,\theta)\in \mathcal{E} \mid (X-r\cos\theta)^2+(r\sin\theta)^2\leq ((1+v)T)^2\}. 
\]
Since the relative velocity of any target with respect to the vehicle is less than or equal to $(1+v)$, the distance $d_1$ of any point on the boundary of $S_T$ from $(X,0)$ is at most $T(1+v)$. This means that the set $S_T$ will be contained in a circle centered at the same position $(X,0)$ of radius $(1+v)T$, i.e., $S_T \subset \bar S_T$. If $T_d$ denotes the minimum amount of time needed to go from vehicle location $(X,0)$ to any target, then, $T_d > T$, if $S_T$ is empty i.e., $\mathbb{P}[T_d>T]=\mathbb{P}[\abs{S_T}=0]$. Moreover,
\begin{equation}
\begin{aligned}
\mathbb{P}[\abs{\bar S_T} = 0]&=\mathbb{P}[\abs{ S_T}=0]\mathbb{P}[\abs{\bar S_T \setminus  S_T}=0] \leq \mathbb{P}[\abs{S_T}=0].\nonumber
\end{aligned}
\end{equation}
Now, consider an infinitesimal area element $\zeta$ of length and width $m$ at $(s,0)$ from the center (cf.~Fig.~\ref{fig:ST} (b)). From Lemma \ref{outstanding demands}, the probability that $\zeta$ is empty will be:
\begin{equation}
\begin{aligned}
\mathbb{P}[\mid \zeta\mid=0] = e^{\frac{-\lambda}{2\pi v}d\theta dr} = e^{\frac{-\lambda}{2\pi v}\frac{m^2}{s}} &\geq e^{\frac{-\lambda m^2}{2\pi v \rho}} = e^{\frac{-\lambda}{2\pi v \rho}\area(\zeta)}, \nonumber
\end{aligned}
\end{equation}
where the inequality comes from the fact that $s \in [\rho,D]$. Since every compact set can be written as a countable union of non-overlapping area elements, the result is true for $\bar S_T$ as well. Therefore,
\[ \mathbb{P}[\mid S_T\mid=0]\geq \mathbb{P}[\mid\bar S_T\mid=0]\geq e^{\frac{-\lambda}{2v \rho}((1+v)T)^2}, \]
and thus, the expectation of $T_d$ can be bounded as
\begin{multline*}
\mathbb{E}[T_d] =  \int_{0}^{\infty} \mathbb{P}[T_d > T] dT = \int_{0}^{\infty}\mathbb{P}[\abs{ S_T}=0] dT \\ \geq  \int_{0}^{\infty} \mathbb{P}[\abs{\bar S_T}=0]dT 
= \frac{\sqrt{\pi}}{2(1+v)} \sqrt{\frac{2v \rho}{\lambda}}
= \frac{1}{1+v}\sqrt{\frac{\pi v \rho}{2\lambda}}. 
\end{multline*}
\end{proof}
We now present the main result of this section.
\begin{theorem}[Upper bound on capture fraction] \label{thm:upperbound}
Any policy $P$ for the RIT problem must satisfy
\[ \subscr{\mathbb{F}}{cap}(P) \leq \min\Bigg\{1, (1+v)\sqrt \frac{2}{v\lambda \pi \rho}\text{ }\Bigg\}. \]
\end{theorem}
\begin{proof}
To service a fraction $c \in (0,1]$ of targets, the service rate of the targets must be greater than the arrival rate \cite{Kleinrock:1975:TVQ:1096491}, i.e., $c\lambda\mathbb{E}[T_d]\leq1.$ Using Lemma \ref{Expectation time}, we obtain this result.
\end{proof}

With this fundamental limit in place, we will now present several policies and examine their performance in comparison to this upper bound.

\section{Low Arrival Rate of Targets}\label{sec:low}
We propose two main policies for the parameter regime of low arrival rates of targets. The FCFS policy is very simple to implement whereas the look ahead (LA) policy requires repeated computation of longest paths.
\subsection{First-Come-First-Serve (FCFS) policy}
According to this policy, the vehicle intercepts the targets in the order in which they arrive. If there are no outstanding targets, the vehicle waits at the center of $\mathcal{E}$ for the targets to appear. The policy is summarized in Algorithm~\ref{algo:fcfs}.
\begin{algorithm}[htpb]
	\DontPrintSemicolon
	\SetAlgoLined
	Assumes that the vehicle is at the center.\\
	\eIf{No outstanding targets in $\mathcal{E}$}{
		Wait for the next target to arrive, \;}
	{Move to intercept the target farthest from the boundary. \;
	}
	Repeat.\;
	\caption{First-Come-First-Serve (FCFS) policy}
	\label{algo:fcfs}
\end{algorithm}

\begin{theorem}[FCFS capture fraction]\label{thm:SAC fraction}
For any $v\in [0,1), \lambda \geq 0$, the capture fraction of the FCFS policy satisfies
\[\subscr{\mathbb{F}}{cap}(\textup{FCFS}) \geq \frac{1}{1+2\lambda \rho}.
\]
\end{theorem}
\begin{proof}
For $\subscr{n}{capt}(t) > 0$ at some $t > 0$, 
\begin{equation}
\begin{aligned}
\subscr{\mathbb F}{cap} 
&= \lim_{t\to +\infty}\sup\mathbb E \left[\frac{1}{1+\frac{n_{esc}(t)}{n_{capt}(t)}}\right]\\
&\geq \left(1+ \lim_{t\to +\infty}\sup\mathbb E \left[\frac{n_{esc}(t)}{n_{capt}(t)}\right]\right)^{-1} \nonumber
\end{aligned}
\end{equation}
where the last inequality comes from an application of Jensen's inequality applied to the convex function, $\frac{1}{(1+x)}$~\cite{Cvetkovski2012}. Thus, by quantifying the expected number of targets that escape per targets captured, we can determine a lower bound on the capture fraction. 

The FCFS policy is difficult to analyze directly. Thus, we provide the lower bound by analyzing an even simpler \emph{Stay-at-Center (SAC) policy} in which the vehicle captures a target just before the target enters the perimeter and returns back to the center. Consider a target $i$ within the capturable set of the vehicle. The time taken by the vehicle to capture the target just before the perimeter and return back to the center will be $2\rho$. Thus, the number of targets that can breach the perimeter will be equal to the number of targets that enter the capturable set while the vehicle is capturing the $i^{th}$ target, i.e., all the targets that are at a distance of $2\rho v$ from the capturable set. Thus, from Lemma~\ref{outstanding demands}, the expected number of targets that will escape is $2\lambda \rho$ and the result follows.
\end{proof}
\subsection{Policies based on Look-ahead}
We now propose and analyze a class of policies in which we constrict the movement of the vehicle such that the vehicle can only move along the perimeter. While such a motion may be sub-optimal in the present context, it allows us to leverage tools from graph theory and adapt our earlier ideas on longest path through the mobile targets to design efficient policies~\cite{smith2009dynamic}.
\begin{definition}[Reachable targets] A target located at $(r,\theta)$ is \emph{reachable} from the vehicle location $(\rho,\phi)$ if
\[ r-\rho \geq v\abs{\theta-\phi} \rho, \text{ for any } v \in [0,1). \]
\end{definition}
\begin{definition}[Reachable set]
The reachable set from the vehicle position $(\rho,\phi) \in \mathcal{E}$ is
\[ R(\rho,\phi) := \{(r,\theta) \in \mathcal{E} : r-\rho \geq v\abs{\theta-\phi} \rho\}.\]
\end{definition}
This means that a target is reachable if it lies in $R(\rho,\phi)$. Next, we define the notion of a reachability graph over the set of radially incoming targets that is inspired out of a similar construction in \cite{smith2009dynamic}.
\begin{definition}[Reachability graph]
A reachability graph of a set of points $\{(r_1, \theta_1), \dots,(r_n, \theta_n)\} \in \mathcal{E}$, is a DAG with vertex set $V :=\{1,...n\}$, and edge set $E$ where, for $j,k \in V$ and $r_j< r_k$, the edge $(j,k)$ is in $E$ if and only if
\[
(r_k-(r_j-\rho),\theta_k) \in R(\rho,\theta_j).
\]
\end{definition}

We now define a policy based on look-ahead wherein the vehicle computes the longest path in the reachability graph of all the targets in $\mathcal{E}$ and captures them on the perimeter. Algorithm \ref{algo:LA} describes the algorithm for Look Ahead policy.

\begin{algorithm}
	\DontPrintSemicolon
	\SetAlgoLined
	Assumes that vehicle is located on the perimeter.\\
	Compute the reachability graph of all the targets in $\mathcal{Q}(0)$ and the vehicle position.\\
	Compute the longest path in this graph, starting from the vehicle position.\\
	Capture targets in the order they appear at the perimeter.\\
	Repeat;
	\caption{Look Ahead (LA) policy}
	\label{algo:LA}
\end{algorithm}

Akin to the rectangular environment considered in \cite{smith2009dynamic}, Algorithm~\ref{algo:LA} is difficult to analyze directly. Instead, we design a \emph{Non-causal} Look Ahead (NCLA) policy. In this policy, at time $0$, the vehicle computes a reachability graph, from its current position using the information of all the future targets that are yet to arrive in $\mathcal{E}$. The vehicle then computes the longest path and captures the targets in the order in which they appear in that path. The algorithm for NCLA policy is formalized in Algorithm \ref{algo:NCLA}.
\begin{algorithm}
	\DontPrintSemicolon
	\SetAlgoLined
	Assumes that vehicle is located on the perimeter.\\
	Compute the reachability graph of all the targets in $\mathcal{Q}(0) \cup \mathcal{Q}_{unarrived}(0)$ and the vehicle position.\\
	Compute the longest path in this graph, starting from the vehicle position.\\
	Capture targets in the order they reach the perimeter.\;
	\caption{Non-causal Look Ahead (NCLA) policy}
	\label{algo:NCLA}
\end{algorithm}\\
The following result provides a guarantee on the performance of the LA policy relative to the NCLA.

\begin{theorem}[LA policy]\label{thm:LA fraction}
If $D-\rho \geq v\pi\rho$, then the capture fraction of the LA policy satisfies 
\[\subscr{\mathbb{F}}{cap}(\textup{LA}) \geq \left(1-\frac{v\pi\rho}{D-\rho}\right)\subscr{\mathbb{F}}{cap}(\textup{NCLA}). \]
\end{theorem}
\begin{proof}[Sketch]
Let the generation of targets begin at $t=0$ and consider two scenarios; (a) Look ahead policy is used by the vehicle, and (b) the vehicle uses the Non-causal look ahead policy (Fig. \ref{fig:proofLA}). Then, akin to the steps in the proof of Theorem \RNum{4}.6 of \cite{smith2009dynamic}, we can compare the number of targets captured in both the scenarios. 

Consider a time instant $t_1$ where the vehicle is computing the path through all outstanding targets $\mathcal{Q}(t_1)$ and denote the path by $\Pi_a(t_1)$. The path in scenario (b), denoted as $\Pi_b(t_1)$, that the vehicle will take through $\mathcal{Q}(t_1)$ be given by $((r_1,\theta_1), (r_2,\theta_2)\dots,(r_m,\theta_m)) \in \mathcal{Q}(t_1).$ The target $(r_1,\theta_1)$ is reachable from $\Pi_b(t_1)$, but might not be reachable from $\Pi_a(t_1)$. The path $\Pi_a(t_1)$ comprises targets: $((r_{n+1},\theta_{n+1}), (r_{n+2},\theta_{n+2}\dots, (r_m,\theta_{m}))$, $n \in \{0,...,m-1\}$, where $(r_{n+1},\theta_{n+1})$ is the highest target that is reachable from $\Pi_a(t_1)$. Thus, the length of $\Pi_a(t_1)$ will satisfy $\mathcal{L}_a\geq m-n$,
where $m$ is the length of the path in (b) because the edge weights of the DAG equal unity. Since the worst-case location for the vehicle is when it is diametrically opposite to the incident target direction, the vehicle can capture any target $i$ if $v\pi \rho + \rho \leq r_i $. Targets at $(r_1,\theta_1),\dots, (r_n,\theta_n)$ must be at least $v \pi \rho+\rho$ from the center.
\begin{figure}[t]
	\centering
	\includegraphics[scale=0.35]{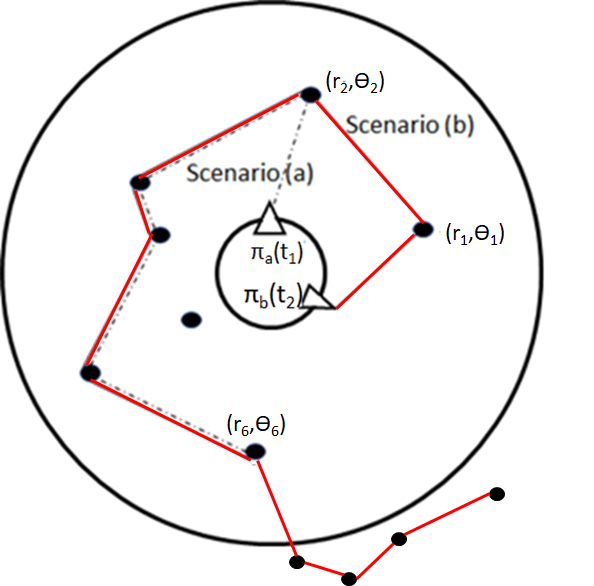}
	\caption{Scenarios for proof of Theorem \ref{thm:LA fraction}. In (a), vehicle visits 5 targets so $\mathcal{L}_a=5$.(b) vehicle visits 6 targets so $m=6$. $(r_2,\theta_2)$ is the highest target on (b) capturable from $p_a(t_1)$ thus $n=1$. The red line shows the path followed by the vehicle in NCLA through arrived and unarrived targets.}
	\label{fig:proofLA} 
\end{figure}
If $N_{tot}$ is the total number of outstanding targets in $\mathcal{E}$ at time $t_1$, then by Lemma \ref{outstanding demands}
\[ \mathbb{E}[n\mid N_{tot}] = N_{tot}\frac{v\pi\rho}{D-\rho}\subscr{\mathbb{F}}{cap}(\textup{NCLA}). \]
Similarly, in scenario (b),
\[ \mathbb{E}[m\mid N_{tot}] = N_{tot}\subscr{\mathbb{F}}{cap}(\textup{NCLA}).\]
Therefore, using the fact that $\mathcal{L}_a \geq m-n$, 
\[\mathbb{E}[\frac{\mathcal{L}_a}{N_{tot}}\mid N_{tot}]\geq\Bigg(1-\frac{v\pi\rho}{D-\rho}\Bigg)\subscr{\mathbb{F}}{cap}(\textup{NCLA}). \]
The ratio $\mathcal{L}_a/N_{tot}$ is the fraction of outstanding targets from the set $\mathcal{Q}(t_1)$ that will be captured in scenario (a) and the ratio does not depend on $N_{tot}$. Thus, by law of total expectation, we get obtain the desired claim.
\end{proof}

In Theorem~\ref{thm:LA fraction}, we established the performance of the LA policy relative to NCLA policy. However, we can also provide an explicit lower bound on the capture fraction for the LA policy, as summarized in the result below. 
\begin{theorem}[Explicit LA capture fraction]\label{thm:LA}
If $D-\rho \geq v\pi\rho$, then the LA policy satisfies
\[\subscr{\mathbb{F}}{cap}(\textup{LA})\geq\frac{1}{\pi\sqrt{\lambda \rho}\erf(\sqrt{\lambda\pi \rho})+e^{-\lambda\pi \rho}},  \]
where $\erf:\mathbb{R}\to[-1,1]$ is the error function.
\end{theorem}
\begin{proof}
The central idea is to construct an invertible transformation between any realization of targets and the vehicle in $\mathcal{E}$ to a rectangular environment akin to the problem
considered in \cite{smith2009dynamic}. We then use the analysis from Theorem IV.8 from \cite{smith2009dynamic} to arrive at this claim.

Consider an environment $\mathcal{E}_1$ that can be constructed by cutting along the radius of $\mathcal{E}$ as illustrated in Fig.~\ref{img:environment transformation}. The transformed environment $\mathcal{E}_1$ is an isosceles trapezoid with the two parallel sides equal to $2\pi D$ and $2 \pi \rho$ respectively. In environment $\mathcal{E}_1$, the transformed targets move from an initial point on the longer side toward the corresponding point on the smaller side so that the time taken by every target to reach the smaller side is the same for all targets. This can be achieved by scaling the speed of each target appropriately as a function of the initial location of the target. $\mathcal{E}_1$ can further be transformed to a rectangular environment $\mathcal{E}_2$ in which the transformed targets move from an initial point on the lower edge toward the upper edge with equal speeds. 
\begin{figure}[t]
	\centering
	\includegraphics[scale=0.35]{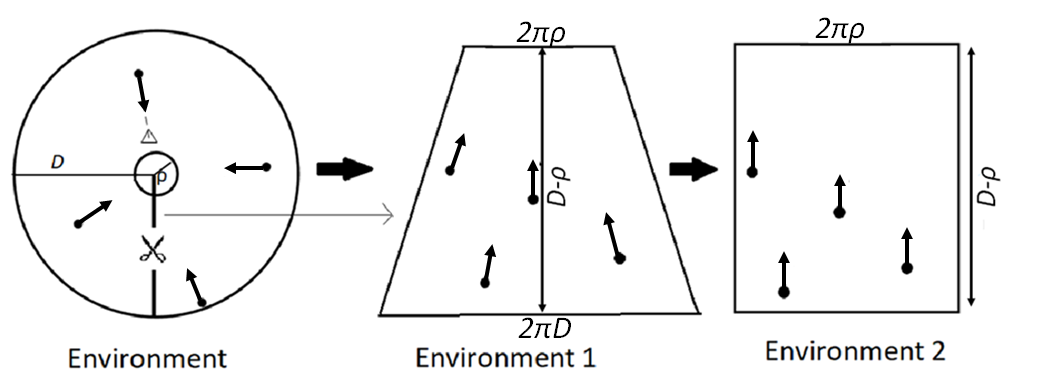}
	\caption{The transformation of the environment from a disk like to a rectangle.}
	\label{img:environment transformation}
\end{figure}

Mathematically, consider a target located at $(r,\theta)$ in $\mathcal{E}$. This means that the target is $r-\rho$ distance away from the perimeter. Thus, in $\mathcal{E}_2$, the target will be $r-\rho$ distance away from the upper edge. Similarly, the target will be at a distance of $\theta \rho$ from one of the sides in $\mathcal{E}_2$. Thus, the location of the transformed target in $\mathcal{E}_2$ will be $(\theta \rho, r-\rho).$ Similarly, if the vehicle is located at $(\rho,\phi)$ and constricted to move on the perimeter, its location after transformation in $\mathcal{E}_2$ will be $\phi \rho$ distance away from one of the side and on the upper edge as in $\mathcal{E}_1$ and captures targets as per Algorithm~\ref{algo:LA}.
Thus, given \emph{any} location of the vehicle on the perimeter and \emph{any} set of locations of targets in $\mathcal{E}$, we can construct a corresponding vehicle location on the upper edge and a set of locations of corresponding targets in $\mathcal{E}_2$ such that the number of targets intercepted by the vehicle in the $\mathcal{E}$ and $\mathcal{E}_2$ is equal. Then, the capture fraction of the LA policy applied to the RIT problem in environment $\mathcal{E}$ is lower bounded by the LA policy applied to the transformed problem in environment $\mathcal{E}_2$. In particular, the LA policy is expected to perform better because of the circular symmetry of $\mathcal{E}$ that allows the vehicle in $\mathcal{E}$ to reach certain targets in a shorter duration. This completes the proof.
\end{proof}	
Note that the capture fraction of LA policy in Theorem~\ref{thm:LA} is independent of $v$. Furthermore, the capture fraction of FCFS policy presented in this section performs well for low arrival rates, i.e., $\mathbb{F}_{cap}(\textup{FCFS})\to 1$ for $\lambda \to 0$, but not in high arrival regime since the upper bound scales as $O(1/\sqrt{\lambda})$. We seek an improved policy in this regime which is the focus of the next section.
\section{High Arrival Rate of Targets}\label{sec:high}
We now introduce a policy based on repeated computation of the EMHP through outstanding radially moving targets. We term this policy as the \emph{Radial Minimum Hamiltonian Path (RMHP) policy}. The key idea is that the number of targets accumulating near the perimeter is high. Thus, the time taken by the vehicle to capture successive targets will be small, and therefore, the vehicle can capture a large number of targets in a single iteration. The vehicle uses solution of the EMHP path (i.e., a path that visits all the points exactly once) to determine the order of the targets to capture. In particular, we consider a constrained EMHP problem which starts at a specific point while visiting all the given set of points \cite{smith2009dynamic}.

The RMHP-fraction policy is defined in Algorithm \ref{algo:RMHP}. In this policy, the vehicle computes the EMHP path through all outstanding targets in $(2\rho,3\rho)$ for all $\theta \in [0,2\pi)$ (Fig.\ref{img:RMHP tour}). The vehicle captures the targets for $\frac{\rho}{v}$ time units and then recomputes the EMHP path through the outstanding targets, allowing the remaining targets in that batch to escape.
\begin{algorithm}
	\DontPrintSemicolon
	\SetAlgoLined
	Assumes that vehicle is located at a distance of $2\rho$ from the origin.\\
	Compute an EMHP path through the outstanding targets located at distance between $2\rho$ and $3\rho$ from the origin.\\
	\eIf{time to travel entire path is less than $\frac{\rho}{v}$}{
	Capture all the outstanding targets by following the computed path.}
	{Capture the targets in the order given by the EMHP.}
	Repeat from step 2.\;
	\caption{RMHP-fraction policy}
	\label{algo:RMHP}
\end{algorithm}
\begin{figure}[t]
	\centering
	\includegraphics[scale=0.4]{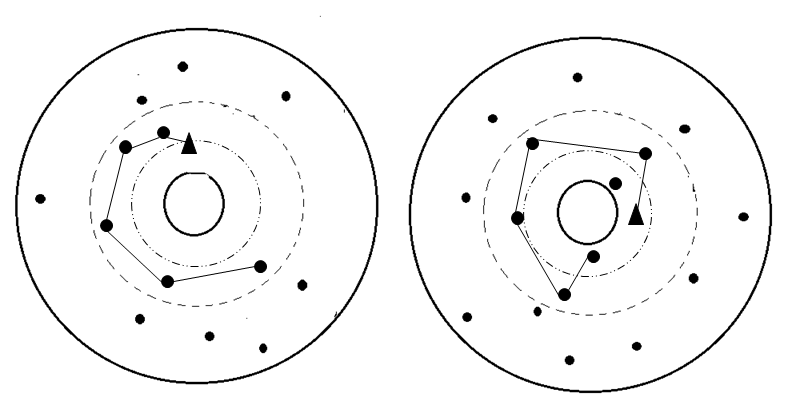}
	\caption{The RMHP-fraction policy. The figure on the left shows the EMHP through all outstanding targets. The figure on the right shows the instant when the vehicle has followed the path through $\rho/v$ time units and allows some targets to escape.}
	\label{img:RMHP tour}
\end{figure}
\begin{theorem}[(RMHP-fraction capture fraction]\label{thm:RMHP}
In the limit as $\lambda \to +\infty$, the capture fraction of the RMHP-fraction policy, with probability one is
\[
\subscr{\mathbb{F}}{cap}\textup{(RMHP-fraction)}\geq \min\Bigg\{1,\frac{1-v}{\alpha \sqrt{v\lambda \rho(1+\sqrt{v})}} \Bigg\},
\]
\end{theorem}
where $\alpha=6\sqrt{2}.$\\
\begin{proof}
Consider the beginning of an iteration of the policy and assume that the duration of the previous iteration was $\rho/v$ time units. At this instant, the vehicle will be $2\rho$ distance away from the center and suppose there are $n$ outstanding targets in an annulus of radii $2\rho$ and $3\rho$. From Lemma \ref{pathbound} and Lemma \ref{Fews}, the length of the tour through these targets can be upper bounded by $\alpha \rho\sqrt{n}/(1-v)$ where, $\alpha=6\sqrt{2}$. As the vehicle can capture targets for at most $\frac{\rho}{v}$ time units, it will capture $cn$ targets, where 
\begin{equation}
    \begin{aligned}
    c&=\min\Bigg\{1,\frac{\rho/v}{(\rho\alpha \sqrt{n})/(1-v)}\Bigg\} =\min\Bigg\{1,\frac{1-v}{v\alpha \sqrt{n}}\Bigg\}.\nonumber
    \end{aligned}
\end{equation}
From Lemma \ref{outstanding demands}, the expected value $\mathbb{E}[n]$ and the variance $\sigma^{2}_{n}$ of the random variable $n$ is $\frac{\lambda \rho}{v}$. Using Chebyshev inequality, $\mathbb{P}[\abs{ n-\mathbb{E}[n]}\geq\gamma]\leq\sigma^{2}_{n}/\gamma^2$, and $\gamma=\sqrt{v}\mathbb{E}[n]$,
\[\mathbb{P}[n\geq(1+\sqrt{v})\mathbb{E}[n]]\leq\frac{1}{v\mathbb{E}[n]}=\frac{1}{\lambda \rho}. \]
Thus, we have
\[c\geq \min\Bigg\{1,\frac{1-v}{\alpha \sqrt{v\lambda \rho(1+\sqrt{v})}} \Bigg\} \]
with probability at least $1-\frac{1}{\lambda \rho}.$ 
\end{proof}
\begin{corollary}[Constant factor of optimality:]
For $v \to 0^+$ and $\lambda \to +\infty$, the capture fraction
\[c\geq \min\Bigg\{1,\frac{1}{\alpha \sqrt{v\lambda \rho}} \Bigg\},
\]
which is within a constant of $12/\sqrt{\pi} \approx 6.77$.
\end{corollary}
Using Lemma~\ref{lem:BHH}, we obtained a better bound on the length and thus a better constant factor of optimality.
\begin{corollary}[Improved constant factor of optimality:]For $v \to 0^+$ and $\lambda \to +\infty$, the capture fraction
\[c\geq \min\Bigg\{1,\frac{1}{\alpha \sqrt{v\lambda \rho}} \Bigg\},
\]
which is within a constant of $3.988/\sqrt{\pi} \approx 2.25$.

\end{corollary}

\section{Simulations}\label{sec:numerics}
We now present results of numerical experiments for the policies analyzed in Sections \ref{sec:low} and \ref{sec:high}. The parameters $D=20$ and $\rho=3$ were kept fixed. The first result compares the FCFS policy to the lower and fundamental bounds. Comparison of the LA policy to the NCLA policy and to the theoretical bounds is shown in the second result. Finally, the last result compares the RMHP-fraction policy to the theoretical bounds.

For the FCFS policy, we simulated 30 runs for each arrival rate. Fig.~\ref{img:FCFS} shows a comparison of the FCFS policy with the theoretical bounds.
\begin{figure}[t]
	\centering
	\includegraphics[scale=0.3]{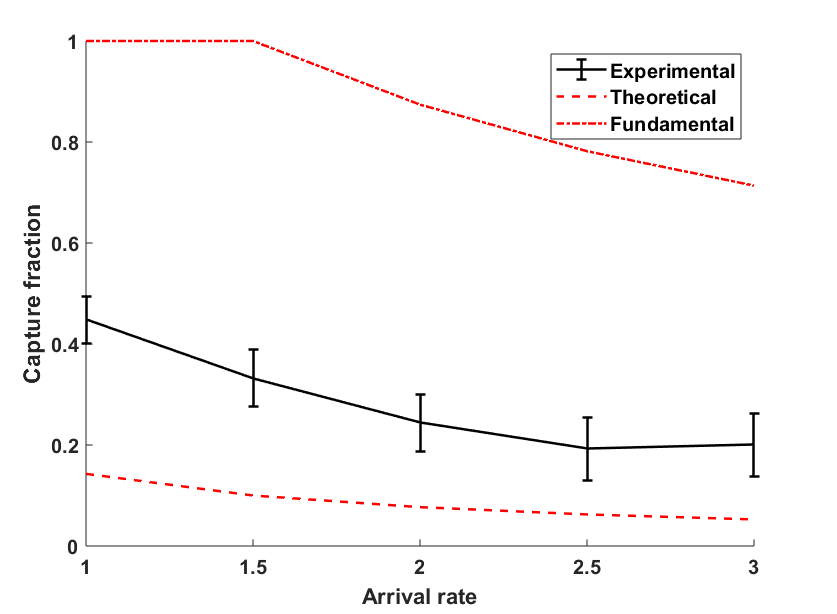}
	\caption{Simulation results for FCFS. Comparison of FCFS policy with fundamental and theoretical bounds for v=0.2. The red dash-dot line shows the fundamental upper bound and red dash curve shows the theoretical bound for FCFS.}
	\label{img:FCFS}
\end{figure}
To simulate the NCLA and LA policies, we implement 5 runs for each value of the arrival rate. Fig.~\ref{img:NCLAsim} shows the comparison of the LA policy with NCLA, fundamental and theoretical bounds.
To simulate the RMHP-fraction policy, the {\ttfamily linkern}\footnote{The TSP solver {\ttfamily linkern} is
  freely available for academic research use at {\ttfamily
    http://www.math.uwaterloo.ca/tsp/concorde/}.} solver was used. For each value of the arrival rate, 5 runs of the policy was taken. The comparison is shown in Fig. \ref{img:RMHP}. Note that the experimental results for RMHP-fraction policy are lower than the theoretical lower bound in Theorem \ref{thm:RMHP}. This is because we have not reached the limit as $v\to 0^{+}$ and $\lambda\to+\infty$. Moreover, we utilize an approximate solution for the EMHP generated by the linkern solver.\\

\begin{figure}[t]
	\centering
	\includegraphics[scale=0.3]{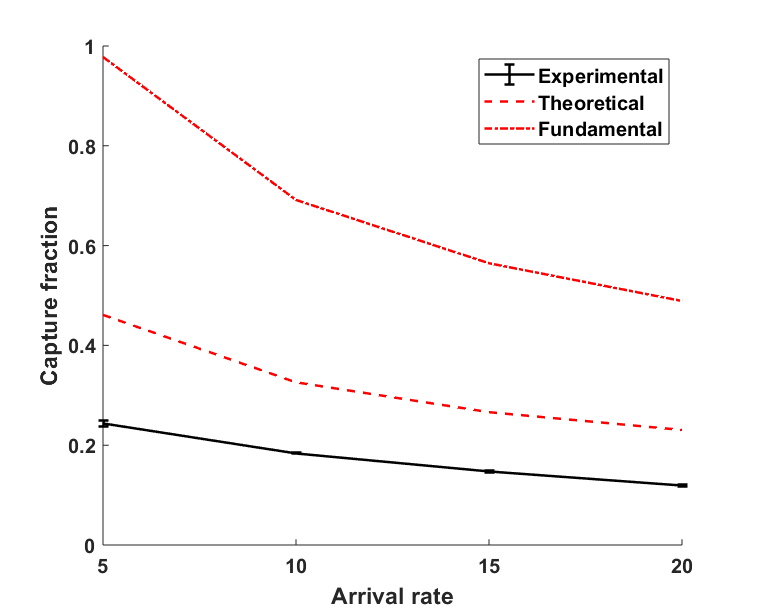}
	\caption{Simulation results for RMHP-fraction policy for v=0.04. The red dash-dot curve is from the policy independent upper bound in Theorem~\ref{thm:upperbound} and the dashed curve shows the theoretical bound. }
	\label{img:RMHP}
\end{figure}
\begin{figure}[t]
	\centering
	\includegraphics[scale=0.3]{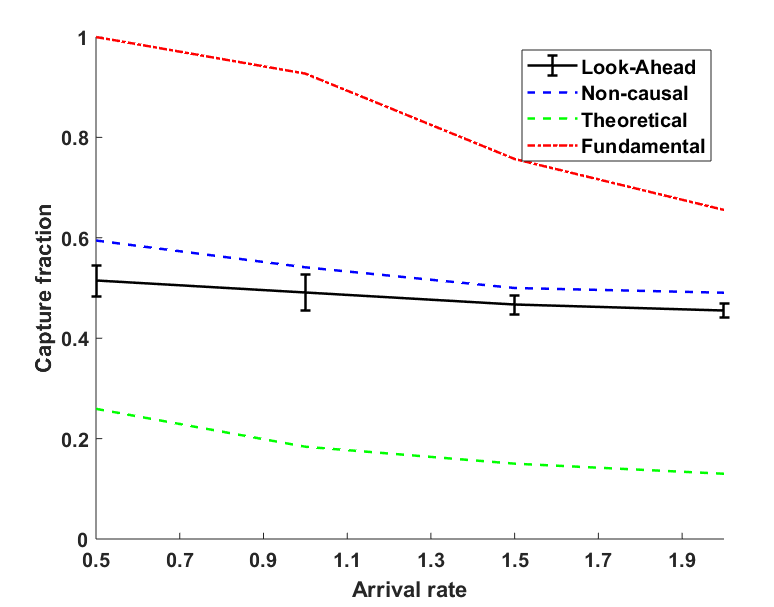}
	\caption{Simulation results for NCLA and LA policy for v=0.8. The red curve shows the fundamental upper bound. Capture fraction of NCLA is shown by the blue curve and the theoretical bound analyzed in \ref{thm:LA fraction} is shown in green color.}
	\label{img:NCLAsim}
\end{figure}
\section{Conclusions and Future Work}\label{sec:conclusion}
This paper introduced the RIT problem, in which a vehicle seeks to defend a perimeter from the radially inward moving targets. We established a policy independent upper bound on the capture fraction of the targets. We then proposed three different policies suitable for low and high arrival rates of the targets and analyzed their respective capture fractions. In the case of low arrival rate of targets, we proposed FCFS, and LA policies. For the latter case, we introduced the RMHP-fraction policy which is within a constant factor of optimal. 

This problem can be extended in many ways. One can consider the case when the speed of the targets is not constant or the targets can maneuver away from the vehicle to escape. Other extensions include multi-vehicle versions of this problem.

\bibliographystyle{IEEEtran}
\bibliography{reference}

\end{document}